\documentclass[
  prl,
  aps,
  reprint,
  superscriptaddress,
  floatfix,
  longbibliography,
  nofootinbib
]{revtex4-2}

\usepackage{amsmath,amssymb,amsthm,mathtools}
\usepackage{bm}
\usepackage{array}
\usepackage{booktabs}
\usepackage{microtype}
\usepackage{tikz}
\usetikzlibrary{arrows.meta}
\usepackage[colorlinks=true,linkcolor=blue,citecolor=blue,urlcolor=blue]{hyperref}

\usepackage{dsfont} 

\usepackage{centernot}
\usepackage{graphicx}
\usepackage{tikz}
\usepackage{amsmath}
\usetikzlibrary{patterns}
\usepackage{thmtools}
\usepackage{thm-restate}
\usepackage[ruled,vlined,linesnumbered]{algorithm2e}

 \usepackage{footnotebackref} 

\usepackage[capitalize,nameinlink]{cleveref}

\usepackage[most]{tcolorbox}  
\allowdisplaybreaks            
\newcommand{\Delx}{\Delta_{XY}}  
\newcommand{\Delq}{\Delta_N}     

\makeatletter

\let\FC@parse\NAT@parse
\def\NAT@parse#1{
  \@ifundefined{fcite@#1}{
    \global\@namedef{fcite@#1}{}
    \Hy@raisedlink{\hypertarget{fcite.#1}{}}
  }{}%
  \FC@parse{#1}
}

\AtBeginDocument{%
  \let\FC@lbibitem\@lbibitem
  \def\@lbibitem[#1]#2{\let\NAT@parse\FC@parse\FC@lbibitem[#1]{#2}}%
}

\def\NAT@anchor#1#2{%
 \hyper@natanchorstart{#1\@extra@b@citeb}%
  \def\@tempa{#2}\@ifx{\@tempa\@empty}{}{%
    \@ifundefined{fcite@#1}%
      {\@biblabel{#2}}%
      {\begingroup\def\@linkcolor{black}\hyperlink{fcite.#1}{\@biblabel{#2}}\endgroup}%
  }%
 \hyper@natanchorend
}
\makeatother

\theoremstyle{plain}

\newtheorem{theorem}{Theorem}

\newtheorem{lemma}{Lemma}

\newtheorem{remark}{Remark}
\newtheorem*{unnumberedtheorem}{Theorem}

\newtheorem{subtheorem}{Theorem}

\makeatletter
\@addtoreset{subtheorem}{theorem}
\@addtoreset{sublemma}{lemma}
\makeatother

\theoremstyle{definition}
\newtheorem{definition}{Definition}

\newcommand{\cL}{\mathcal L}
\newcommand{\cQ}{\mathcal Q}
\newcommand{\cG}{\mathcal G}
\newcommand{\NS}{\mathcal{NS}}
\newcommand{\R}{\mathbb R}
\newcommand{\ip}[2]{\left\langle #1,#2\right\rangle}

\newcommand{\conv}{\operatorname{conv}}

\newcommand{\ket}[1]{\left|#1\right\rangle}

\newcommand{\1}{\mathds{1}} 

\usepackage[normalem]{ulem}
\usepackage{cancel}
\newif\ifshowchanges
\showchangestrue
\definecolor{revisionblue}{RGB}{0,70,190}
\DeclareRobustCommand{\remove}[1]{%
  \ifshowchanges{\color{red}\hypersetup{allcolors=red}%
    \ifmmode\cancel{#1}\else\sout{#1}\fi}\fi}

\newcommand{\revisiondisplay}[1]{%
  \ifshowchanges\par\addvspace{\smallskipamount}\noindent
  \begingroup\setbox0=\hbox{$\displaystyle #1$}%
  \hfill\ifdim\wd0>\linewidth
    \resizebox{\linewidth}{!}{\usebox0}\else\usebox0\fi\hfill\mbox{}%
  \par\endgroup\addvspace{\smallskipamount}\fi}
\newcommand{\revisioninline}[1]{%
  \begingroup\setbox0=\hbox{#1}%
  \ifdim\wd0>.88\linewidth\resizebox{.88\linewidth}{!}{\usebox0}%
  \else\usebox0\fi\endgroup}
\begin{document}

\title{Bell nonlocality does not imply violation of nonlocal games}

\author{Ricardo Faleiro}
\affiliation{Instituto de Telecomunicações, Av. Rovisco Pais, 1049-001, Lisboa, Portugal}

\author{Lucas E. A. Porto}
\affiliation{Sorbonne Université, CNRS, LIP6, F-75005 Paris, France}

\author{Marco Túlio Quintino}
\affiliation{Sorbonne Université, CNRS, LIP6, F-75005 Paris, France}

\author{Flavien Hirsch}
\affiliation{Instituto de Telecomunicações, Av. Rovisco Pais, 1049-001, Lisboa, Portugal}

\begin{abstract}Nonlocal games offer a conceptually clear way to reason about Bell nonlocality. However, we show that as originally defined, and most commonly used, \textit{i.e.} with a deterministic predicate, they cannot be complete witnesses of nonlocality. Concretely, we prove the existence of nonlocal behaviours which do not violate any such nonlocal game. Furthermore, we also show that quantum mechanics can produce these game-classical nonlocal behaviours. We prove that this phenomenon exists in every bipartite scenario $(X,Y,A,B)$ with $X,Y\geq2$ and $A,B\geq4$.
\end{abstract}

\maketitle

\emph{Introduction.---}
Bell nonlocality shows that quantum systems can produce correlations that no natural classical common cause can explain. These correlations challenge our understanding of reality and enable various tasks impossible classically, for instance cryptographic protocols whose security does not depend on trusting the internal operation of the devices used.

In the simplest bipartite Bell experiment, Alice and Bob receive inputs \(x\) and \(y\) and, without communicating, produce outputs \(a\) and \(b\). The numbers of possible inputs are denoted by \(X\) and \(Y\), and the numbers of possible outputs by \(A\) and \(B\); the scenario is abbreviated \((X,Y,A,B)\). If a classical common cause, formalized as a Local Hidden-Variable (LHV) model, can reproduce the resulting input-output correlations (called a behaviour), then that behaviour is local; otherwise, it is nonlocal. In a Bell scenario, nonlocality is exhibited by breaking the classical correlation bounds (attained from LHVs) of so-called Bell inequalities. Quantum theory predicts, and experiments confirm, the existence of nonlocal behaviours that violate Bell inequalities~\cite{Bell64,CHSH,Brunner14,Hensen2015}. 

Bell inequalities provide a clear way to exhibit nonlocality. Moreover, mathematically, the set of local behaviours \(\mathcal{L}\) forms a polytope, fully characterized by finitely many Bell inequalities, which define the boundary of the local set \cite{froissart1981constructive,fine1982hidden}. Thus, for a fixed behaviour, a proof of nonlocality becomes a well-defined membership problem. 

Nonlocal games turn such tests of Bell nonlocality into interactive proofs with a transparent operational meaning~\cite{cleve2004consequences}. Instead of Alice and Bob accumulating,  across many runs of their experiment, the overall input-output statistics observed, a nonlocal game instead describes an
interaction turn-by-turn. In a single turn, a referee draws questions $x$ and $y$, sends one to each of two separated non-communicating
players, receives their answers $a$ and $b$, and declares whether they win or lose based on a predicate $V(a,b|x,y)$. The game is played for many turns. 
One aspect which benefits this description is the conceptual clarity provided by the interactive picture and the operational significance of the figures of merit (winning probabilities of the game) used to show nonlocality. 
This approach is standard in relating Bell nonlocality with notions of complexity
theory~\cite{cleve2004consequences,buhrman2010nonlocality,palazuelos2016survey}, and is particularly natural for the statistical analysis of Bell tests, simplifying the treatment of some loopholes like the detection loophole and the memory loophole, as well as providing physically meaningful notions of quantum-to-classical gap~\cite{araujo2020bell,Barrett_2002_memory_loophole,Hensen2015}. 

In addition to their interactive character, the original and still widely adopted convention presupposes that such nonlocal games have deterministic predicates, that is, \(V(a,b|x,y) \in \{0,1\}\). Nevertheless, some works extended the definition of nonlocal games to allow ``probabilistic predicates", that is, upon receiving inputs $x$ and $y$ and producing outputs $a$ and $b$, Alice and Bob win the game with probability $V(a,b|x,y)\in [0,1]$ \cite{palazuelos2016survey}. Furthermore, it was also later shown, that under this broader convention one recovers every possible Bell inequality \cite{araujo2020bell}. Throughout this work, nonlocal games have \textit{deterministic predicates} as originally defined~\cite{cleve2004consequences}.

At this stage, a natural question one can ask is whether nonlocal games can detect any nonlocal correlation.
Granted, the set of Bell inequalities that can be written as nonlocal games is more restricted but, perhaps, nonlocal
games could still be rich enough to detect every nonlocal behaviour. For instance, this is the case in the simplest scenario i.e. \((2,2,2,2)\),  since every nonlocal behaviour therein is detectable by the Clauser-Horne-Shimony-Holt (CHSH) inequality~\cite{CHSH}, which, as is well known, can be written as a nonlocal game. But indeed, not only the CHSH, the vast majority of Bell inequalities that have been studied and experimentally tested are equivalent to nonlocal games, whether they have been explicitly represented as nonlocal games, or not. Just to name a few more examples: the Greenberger–Horne–Zeilinger (GHZ)~\cite{GHZ},  Collins–Gisin–Linden–Massar–Popescu (CGLMP)~\cite{CGLMP} , \(I_{3322}\)~\cite{I3322} inequalities; all binary-output full-correlator inequalities, the recently introduced Unanimous inequalities \cite{smells26}; on top of, obviously, all inequalities which are already explicitly represented as games,  XOR games~\cite{cleve2004consequences}, unique
games~\cite{kempe2010unique}, pseudo-telepathy games such as the magic
square~\cite{brassard2005pseudotelepathy}, and binary constraint-system
games~\cite{cleve2014constraint}.

This means that, so far, despite being more mathematically constrained, nonlocal games have, to a large extent, been implicitly made to be operational equivalents of Bell inequalities. Thus, we are interested in knowing whether this is warranted, which would mean that nonlocal games can detect all possible nonlocal behaviours.

Here we show that there exist nonlocal behaviours that do not violate any deterministic-predicate nonlocal game in the same Bell scenario. Furthermore, quantum mechanics can produce such behaviours.

We define the set of behaviours that satisfy every nonlocal game, \(\mathcal{G}\), show that it is a polytope, characterize its facets  and show it to be in general larger than the local polytope  \(\mathcal{L}\).

\emph{Behaviours, Bell inequalities, and nonlocal games.---}
Consider a bipartite Bell scenario $(X,Y,A,B)$.  Alice and Bob receive inputs
$x\in\{0,\ldots,X-1\}$ and $y\in\{0,\ldots,Y-1\}$ and return outputs
$a\in\{0,\ldots,A-1\}$ and $b\in\{0,\ldots,B-1\}$.  A behaviour
$p=\{p(ab|xy)\}_{abxy}$ contains the joint probabilities of Alice and Bob outputting $a,b$ given $x,y$ and thus is normalized and nonnegative.  It is no-signalling
when Alice's marginal  is independent of $y$, i.e., $\sum_b p(ab|xy)=p_A(a|x)$ is independent of $y$, and Bob's marginal  is independent of $x$, i.e., $\sum_a p(ab|xy)=p_B(b|y)$~\cite{Brunner14}.  We write $\NS$ for the polytope of all such
behaviours and take its affine hull as the ambient space throughout.  A
behaviour is Bell local, $p\in\cL$, if it
admits a local hidden-variable decomposition
\begin{equation}
 p(ab|xy)=\sum_\lambda \pi(\lambda)
 p_A(a|x,\lambda)p_B(b|y,\lambda),
 \label{eq:local-model}
\end{equation}
where $\pi(\lambda)$ is a probability distribution and $p_A$ and $p_B$ are
normalized conditional response distributions.  We denote by $\cL$ the set
of all behaviours admitting the decomposition in
Eq.~\eqref{eq:local-model}. Quantum behaviours have the form
$p(ab|xy)=\operatorname{tr}[\rho(A_{a|x}\otimes B_{b|y})]$, where \(\rho\) is a quantum state (represented by a positive semi-definite  matrix with unit trace) and \(A_{a|x},B_{b|y}\)  positive operator-valued measures (POVMs); their convex set is
denoted by $\cQ$, with $\cL\subseteq\cQ\subseteq\NS$. For finite scenarios, $\cL$ is a polytope, and its facet-defining
inequalities are the tight Bell
inequalities~\cite{froissart1981constructive,pitowsky1989book}. We show the characterization of this set as a polytope in Appendix~\ref{app:definitions}. 

A Bell functional is a real tensor $M$ acting on a behaviour as
\begin{equation}
 \ip{M}{p}=\sum_{a,b,x,y}M_{ab|xy}p(ab|xy). 
\end{equation}
Its local bound is defined as
$\beta_{\mathrm L}(M)=\max_{p\in\cL}\ip{M}{p}$.  The pair composed by the functional and its corresponding local bound, typically represented as
$\ip{M}{p}\leq\beta_{\mathrm L}(M)$ is a Bell inequality. Below, we define the class of Bell functionals which are nonlocal games, as well as the set of \textit{game-classical} behaviours, $\cG$.

\begin{definition}[Nonlocal game~\cite{cleve2004consequences}]  \label{def:NLgame}
  A Bell functional $G$ is a nonlocal game if it can be  written as \(
 G_{ab|xy}=\mu(x,y)V(a,b|x,y)\), 
where \(\mu(x,y)\) is a probability distribution over the inputs, and $V(a,b|x,y)\in\{0,1\}$ a (deterministic) predicate.  Its associated local bound is also called the classical value of the game 
\(\nonumber
 \omega_{\mathrm L}(G)=\max_{p\in\cL}\ip{G}{p}\). 
\end{definition}

\begin{definition}
The set of no-signalling behaviours that do not violate any nonlocal game is defined as
\begin{align}
 \cG:=\bigl\{p\in\NS:\ip{G}{p}\leq\omega_{\mathrm L}(G)
 \quad\forall G\bigr\}.
 \label{eq:game-polytope}
\end{align} We call \(\cG\) the game-classical  set.
\end{definition}

Every nonlocal game is by definition a Bell inequality, so we have the direct inclusion $\cL\subseteq\cG$. However, equation \eqref{eq:game-polytope} provides a characterization of $\cG$ using infinite (linear) inequalities. Thus, although it is straightforward to check that $\cG$ defines a convex set, we prove here that it is also a polytope. Furthermore, we also show a constraint on its facet-defining inequalities. 

\refstepcounter{theorem}

\begin{subtheorem}
\label{thm1:game-polytope}
The set $\cG$ is a polytope. 
\end{subtheorem}
\begin{subtheorem}
\label{thm2:game-polytope}
All facets of $\cG$ admit a nonlocal game representation.
\end{subtheorem}

\cref{thm1:game-polytope} and \cref{thm2:game-polytope} are proven in Appendix \ref{app:G is a polytope with game facets}.

A main question now is whether the set $\cG$ is equal to the local polytope $\cL$. If this were the case, then every nonlocal behaviour would violate some nonlocal game. 
In the $(2,2,2,2)$ Bell scenario, the local polytope has only one class of non-trivial facet-inequality, namely the CHSH inequality.
Since CHSH does admit a nonlocal-game description~\cite{Brunner14}, we have that $\cG_{(2,2,2,2)} = \cL_{(2,2,2,2)} $ .
If we go to slightly higher scenarios, that is $(2,2,3,3)$ and $(3,3,2,2)$, we find again that all the facets of $\cL$ admit a nonlocal-game representation~\cite{araujo2020bell}, and thus the set $\cG$ in these scenarios is also exactly $\cL$.

In the following, we prove that for more complex scenarios, such as the scenario $(4,4,2,2)$, the equality $\cL = \cG$ breaks. We will show that this condition fails by showing explicit examples of facet-inequalities that are game-inequivalent.

We now state our main finding in the following theorem. \cref{fig:sets} visually summarizes its results. 

\begin{theorem}[Main result]
\label{thm:main}
In general, the local set is a strict subset  of the game-classical set: \begin{equation}
    \cL\subsetneq\cG.
    \label{eq:inclusion}
\end{equation} Namely, given a bipartite Bell scenario $(X,Y,A,B)$, the strict inclusion above holds for
$X,Y\geq2$ and $A,B\geq4$ and also for the particular case $(4,4,2,2)$.  This establishes the existence of nonlocal game-classical  behaviours.

Furthermore, quantum mechanics can produce such behaviours. Thus, in the mentioned scenarios we also have that
\begin{equation}
 \cQ\cap(\cG\setminus\cL)\neq\varnothing.
 \label{eq:quantum-separation-4422}
\end{equation}
\end{theorem}

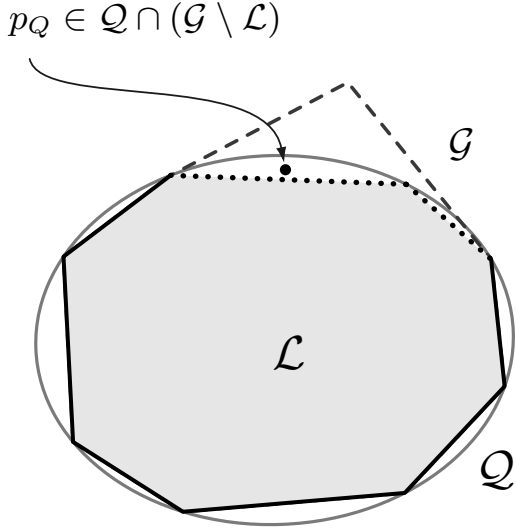
\begin{figure}[t]
\centering
\scalebox{1.5}{%
  \begin{tikzpicture}[
 x=1cm,y=1cm,line cap=round,line join=round,
 qline/.style={draw=black!52,line width=0.8pt},
 gline/.style={draw=black!78,line width=0.9pt,dashed},
 lline/.style={draw=black,line width=1pt},
 dotfacet/.style={draw=black,line width=1.35pt,dash pattern=on 0pt off 3pt,line cap=round},
 small/.style={font=\scriptsize},tinytext/.style={font=\tiny},
 statement/.style={font=\scriptsize,align=center}]
\draw[qline,rotate around={1.98:(2.154,2.928)}] (2.154,2.928) ellipse [x radius=2.104cm,y radius=1.627cm];
\draw[gline] (0.379344,2.029668)--(0.294976,3.663645)--(1.239474,4.380835)--(2.80,5.20)--(4.061348,3.640703)--(4.180421,2.519653)--(3.300410,1.579462)--(1.347448,1.413792)--cycle;
\fill[black!10]
 (0.379344,2.029668)--(0.294976,3.663645)--
 (1.239474,4.380835)--(3.314268,4.301545)--
 (4.061348,3.640703)--(4.180421,2.519653)--
 (3.300410,1.579462)--(1.347448,1.413792)--cycle;
\draw[lline] (4.061348,3.640703)--(4.180421,2.519653)--(3.300410,1.579462)--(1.347448,1.413792)--(0.379344,2.029668)--(0.294976,3.663645)--(1.239474,4.380835);
\draw[dotfacet] (1.239474,4.380835)--(3.314268,4.301545);
\draw[dotfacet] (3.314268,4.301545)--(4.061348,3.640703);
\node[font=\large] at (2.28,2.82) {$\mathcal L$};
\node[font=\normalsize] at (3.80,4.65) {$\mathcal G$};
\node[font=\large] at (4.10,1.78) {$\mathcal Q$};
\fill (2.25,4.43) circle[radius=1.4pt];
\node[font=\bfseries,anchor=west] (membership) at (-0.30,5.75)
 {$p_Q\in\cQ\cap(\cG\setminus\cL)$};
\draw[line width=0.45pt,black!90,
 -{Latex[length=4pt,width=2.6pt]},shorten <=2pt,shorten >=1.7pt]
 ([xshift=8pt]membership.south west)
 .. controls (0.10,4.95) and (2.08,5.32) .. (2.25,4.43);
\end{tikzpicture}
}
\caption{\label{fig:sets}
The local polytope
$\cL$ lies inside both the quantum set $\cQ$ and the game-classical  polytope $\cG$. The dotted facets of $\cL$ represent Bell inequalities which are not game-equivalent. The dashed facets of $\cG$ represent the additional nonlocal games needed to enclose $\cL$.  
Our result proves
$\cQ\cap(\cG\setminus\cL)\neq\varnothing$, since  $p_Q$ is a quantum realizable game-classical  nonlocal behaviour.
That is, $p_Q$ is quantum and Bell nonlocal, but remains in $\cG$. (This is merely schematic since $\cG$ is not characterized in general and  the relationship between $\cG$ and $\cQ$ is not known. It might be that $\cG$ is entirely contained in $\cQ$, or not.) }
\end{figure}

The proof of  \cref{thm:main} follows from the analysis of the facets of the local polytope, and specially how they behave under no-signalling transformations. Below, we define and characterize this class of equivalence of nonlocal games. 

\emph{Defining and finding game-inequivalent facets.---} Early work studying the relationship between Bell inequalities and  nonlocal games found that not all Bell functionals could be written as a  nonlocal games, under translation and scaling of the coefficients~\cite{silman2008relation}. 
In \cref{prop:Bellgame} of Appendix~\ref{app:definitions}, we prove that a Bell functional is equivalent to a nonlocal game under output-independent block translations and a common positive scaling if and only if every input block contains at most two distinct coefficient values, in similar vein to what was shown in Appendix B Ref.~\cite{araujo2020bell}.  It should be stressed that this type of game-equivalence does not consider no-signalling transformations, which are on par with positive scaling and output-independent block translations as free operations in the no-signalling space---the  relevant space in Bell scenarios. \footnote{Positive scaling and block translations can also be considered as free operations in the normalized signalling space; they are, in some sense, less physically motivated by the Bell scenario than no-signalling transformations, which  can be thought as a gauge freedom specific  to the particular physical setup considered. In the signalling space, no-signalling transformations would not be free.} 

We therefore allow arbitrary linear combinations of no-signalling constraints in addition to positive scaling and block translations~\cite{rosset2014classifying,araujo2020bell}. Two inequalities are equivalent when their functionals differ by a positive scaling and a constant on the normalized no-signalling affine space, with their bounds transformed accordingly. This gives the following definition.

\begin{definition}[Game equivalence]
\label{def:Bellequiv}
A Bell inequality with associated functional \(M_{ab|xy}\) is \emph{game-equivalent} if there exists a nonlocal game \(G_{ab|xy}\), such that
\begin{align}
  M_{ab|xy} = c\,G_{ab|xy} + \Gamma_{ab|xy},
\end{align}
where \(c>0\) is a positive scaling and  $\Gamma$ is a Bell functional which is constant on all no-signalling behaviours, i.e., \(\ip{\Gamma}{p_{NS}}= \gamma\) for all \(p_{NS}\in\NS\).

\end{definition}

In Appendix~\ref{app:definitions} we give a practical characterization of $\Gamma_{ab|xy} $.

Using the previous definition, recently, Araújo et al.~\cite{araujo2020bell}, considering arbitrary no-signalling transformations constructed an example of Bell functional that cannot be transformed into any nonlocal game. This example, which was explicitly designed to show game-inequivalence, does not detect nonlocal behaviours and, as far as the authors are aware, no other examples explicitly showing game non-equivalence are known.

In Appendix~\ref{app:MILP}, we show that deciding whether a Bell inequality is equivalent to a nonlocal game can be formulated as a mixed-integer linear program (MILP). This formulation allows us to use standard numerical solvers to search systematically for game representations. 
In case the MILP finds an explicit game representation for the given inequality, this can be independently verified in an analytical way.
On the other hand, due to numerical imprecision, if a MILP solver fails to find such a representation, we do not have an analytical certificate proving the game-inequivalence of the given inequality.
We therefore use the MILP both to find game representations and to identify candidates whose game-inequivalence we subsequently prove analytically. We also notice that MILP were previously used in the context of Bell inequalities for finding binary forms of inequalities Ref.~\cite{denis}. In our language, this can be viewed as nonlocal games with uniform prior.

The complete $(4,4,2,2)$ catalogue contains 174 nontrivial facet classes in
addition to positivity~\cite{zambrini2019complete}.  Our search with the MILP
returned game representatives for 173 of them and singled out one inequality, that we call
$J_\star$.  In compact
Collins--Gisin notation~\cite{I3322} it is
\begin{equation}
\begin{split}
J_\star={}&(3,3,3,3;\,3,-1,-1,-2;\,1,1,-3,2;\,1,-3,1,2)\\
&\mid(-5,1,0,0)\mid(-5,0,0,-3)\leq1.
\end{split}
\label{eq:Jstar-CG}
\end{equation}

The three groups separated by vertical lines multiply, respectively, $p(00|xy)$, Alice's marginals, and
Bob's marginals. In Appendix~\ref{app:J40}, the inequality is given in full notation together with the analytical proof that it cannot be game-equivalent in the \((4,4,2,2)\) scenario. This allows us to formulate the following lemma, proven in Appendix \ref{app:J40}: 
\begin{restatable}{lemma}{lemmaFourFourTwoTwo}
\label{lemma:4422}
In scenario $(4,4,2,2)$, out of the 174 facet Bell inequalities, only one,
$J_\star$, is game-inequivalent in that scenario.
\end{restatable}

That is, since $\cG$ and $\cL$ are both polytopes and at least one facet of $\cL$ is not a facet of $\cG$, we can conclude that $\cG_{(4,4,2,2)} \neq \cL_{(4,4,2,2)}$, and thus the game-classical set and the local set differ for at least one scenario. Next, we show that in fact for any bipartite scenarios with at least $2$ inputs and $4$ outputs per party,  $\cG \neq \cL$.

\emph{Game-inequivalence preserved under lifting.---}
Here, the proof that $J_\star$ is game-inequivalent explicitly uses the number of inputs of the scenario, since it leverages no-signalling transformations to show a contradiction. In fact, a general
construction presented in Araújo et al.~\cite{araujo2020bell}, shows how   every binary-output functional (such as $J_\star$), can be represented as nonlocal game in a scenario where the inputs are doubled for Alice and Bob. Thus, though $J_\star$ is not game-equivalent in \((4,4,2,2)\) it will necessarily be so in \((8,8,2,2)\). In fact,
  $J_\star$ already has a nonlocal game representative in $(5,5,2,2)$, found through the MILP \cite{github_game_vs_inequ}. Thus, even if $J_\star$ proves that the inclusion is strict \(\cL\subsetneq\cG\) in the particular case of \((4,4,2,2)\), it cannot be leveraged to show that the inclusion is strict in higher input scenarios. To remove this limitation, we move away from scenarios with binary outputs (for which this limitation will always be true) and turn to the $(2,2,4,4)$ scenario. 
    A MILP search over the 33 classes of facet inequalities of the scenario's known (and possibly complete) catalogue~\cite{jesus2023tight}, resulted in two candidates, denoted $F^1_{\star}$ and $F^2_{\star}$. We give them both here in Collins-Gisin notation: 

\begin{equation}
\begin{split}
F^1_{\star}={}&(3,2,1,2,2,1,1,1,1;\,3,2,1,2,1,1,1,1,0;\\
&\;\;3,2,1,2,1,1,1,1,0;\,-3,-2,-1,-2,-1,0,-1,0,0)\\
&\mid(-3,-2,-1;\,0,0,0)\mid(-3,-2,-1;\,0,0,0)\leq 0.
\end{split}
\label{eq:F1-CG}
\end{equation}

\begin{equation}
\begin{split}
F^2_{\star}={}&(3,2,2,2,2,1,2,1,0;\,3,2,1,2,2,1,2,2,2;\\
&\;\;3,2,2,2,2,2,1,1,2;\,-3,-2,-2,-2,-2,-1,-2,-1,0)\\
&\mid(-3,-2,-2;\,0,0,0)\mid(-3,-2,-2;\,0,0,0)\leq 0.
\end{split}
\label{eq:F2-CG}
\end{equation}
As for $J_{\star}$, the three groups separated by vertical lines multiply, respectively, $p(ab|xy)$, Alice's marginals $p_A(a|x)$, and Bob's marginals $p_B(b|y)$, with $a,b\in\{0,1,2\}$ and $x,y\in\{0,1\}$. In the first group, the blocks separated by semicolons correspond to $(x,y)=(0,0),(0,1),(1,0),(1,1)$; within each block the entries are ordered $ab=00,01,02,10,\dots,22$. In the marginal groups, semicolons separate $x=0$ from $x=1$ (respectively $y$).

Both are shown to be game-inequivalent in Appendix~\ref{app:F26F29}, using a proof strategy also found in Araújo et al.~\cite{araujo2020bell},  wherein one uses just a single output block of the functional, corresponding to a unique input configuration, which in effect means the proof ignores the number of inputs and works for any extra number of outputs. This and the explicit nonlocal game descriptions for the other facet inequalities in \((2,2,4,4)\) gives the proof for the following Lemma:

\begin{restatable}{lemma}{lemmaTwoTwoFourFour}
\label{lemma:2244}
In scenario $(2,2,4,4)$, out of the 33 known facet Bell inequalities, only two,
$F^1_{\star}$ and $F^2_{\star}$, are game-inequivalent in that scenario.
Furthermore, neither becomes game-equivalent under any standard input or
output lifting.
\end{restatable}

The lemmas above establish that $\cG \neq \cL$. However, they do not directly provide examples of behaviours in $\cG$ but not $\cL$. Next, we discuss approaches to find such examples.

\emph{Existence of a nonlocal quantum behaviour in $\cG$.---}
Combining  \cref{thm1:game-polytope} and \cref{thm2:game-polytope}, together with the fact that $J_\star$ is a facet of $\cL$ but not of $\cG$, we can conclude that the $(4,4,2,2)$ scenario contains behaviours which are nonlocal, but do not violate any game. Indeed, if there were no such behaviours, we would have $\cL = \cG$ in this scenario, and thus $J_\star$ would also be a facet of $\cG$, which cannot be the case as we have proven it is not game-equivalent.  

The behaviour-membership MILP in Eq.~\eqref{eq:MILP_behaviour}, introduced
below, provides a numerical way to search for nonlocal behaviours in
$\cG$. Appendix~\ref{app:quantumNL_in_G} gives a numerical example
and proves analytically that every quantum-violated, game-inequivalent
facet of $\cL$ yields quantum behaviours in $\cG\setminus\cL$.
Appendix~\ref{app:quantum_violations} supplies the facet-violating quantum
witnesses used in that construction. We formalize the geometric implication
in the following lemma.

\begin{lemma}
\label{lemma:qpoint}
A game-inequivalent facet of \(\cL\) that admits a quantum violation implies the existence of nonlocal quantum behaviours in $\cG$.
\end{lemma}

\cref{lemma:qpoint} (whose proof is in Appendix \ref{app:quantumNL_in_G}) shows that, given a scenario, if there is at least one quantum-violated game-inequivalent facet inequality then there is a game-classical quantum behaviour violating this facet. 
In a nutshell, let $p_\star$ lie in the relative interior of a game-inequivalent facet of $\mathcal{L}$. Since $\mathcal{L}\subseteq\mathcal{G}$
and every facet of $\mathcal{G}$ is game-equivalent, $p_\star$
lies in the interior of $\mathcal{G}$. If $q$ is a quantum
behaviour violating that local facet, then
\begin{equation}
    p_t=(1-t)p_\star+tq
\end{equation}
is quantum and nonlocal for every $0<t\leq 1$, while remaining
in $\mathcal{G}$ for sufficiently small $t>0$.

In Appendix \ref{app:quantum_violations} we give examples of such quantum violations for \(J_\star\), \(F^1_\star\) and \(F^2_\star\). This is now enough to prove our main result that:
 \begin{equation*}
 \cQ\cap(\cG\setminus\cL)\neq\varnothing.
\end{equation*}

Indeed, first, from  \cref{lemma:4422}, it directly follows that $\cL\subsetneq\cG$ holds for scenario \((4,4,2,2)\). Second,   \cref{lemma:2244}  provides us with lifting-stable game-inequivalent inequalities and, since either inequality will appear as a lifted facet inequality of \(\cL\) in any higher scenario, since input and output lifting preserve the inequalities as facets~\cite{pironio2005lifting}, we show that the strict inclusion \(\cL\subsetneq\cG\) is true, as long as, Alice and Bob have, each, at least two inputs and four outputs. This precisely proves that $\cL\subsetneq\cG$ holds for scenario \((X,Y,A,B)\), for \(X,Y\geq 2\) and \(A,B\geq 4\).

All in all, we have that: i)  In scenario \((4,4,2,2)\), $J_{\star}$ is a game-inequivalent facet that admits a quantum violation; ii) Every scenario with $X,Y\geq2$ and $A,B\geq4$ contains a lifted version of an $F_{\star}$ inequality, that is, a game-inequivalent facet that admits a quantum violation; iii) any quantum-violated game-inequivalent facet of $\cL$ implies the existence of quantum nonlocal behaviours in $\cG$ . Thus, \cref{thm:main} is proven in full.

\emph{The $\cG$ membership problem.---} Here we define behaviour-membership in $\cG$ as an optimization problem. The idea is to write the maximization of the score of a nonlocal game applied on the target behaviour, over all possible games, minus the bound of that nonlocal game. That is, we try to maximize the quantum-to-classical gap for that behaviour. 

Consider the following optimization problem, defined for an inputted fixed $p(ab|xy)$:

\begin{equation}\label{eq:MILP_behaviour}
\begin{aligned}
c^* = \max_{G,\,V,\,\mu, \beta}\;\;  \sum_{abxy} G_{ab|xy}\, p(ab|xy) - \beta \\
\text{s.t.}\;\; \sum_{abxy} G_{ab|xy}\, D^{(k)}(ab|xy) \leq \beta \quad \forall k, \\
G_{ab|xy} = \mu(x,y) V(a,b|x,y) \\
 0 \leq \mu_{xy} \leq 1 , \; \sum_{x,y} \mu(x,y) = 1 \\
 V_{ab|xy} \in \{0,1\} 
\end{aligned}
\end{equation}
where $D^{(k)}(ab|xy)$ are the local deterministic behaviours of that scenario, i.e. the vertices of $\cL$. 

This optimization problem per se is nonlinear, since we impose $G_{ab|xy} = \mu(x,y) V(a,b|x,y)$. However, we show in Appendix \ref{app:G_membership_milp} that we can linearize this condition and arrive at a well-defined MILP. 
 
This allows a characterization of $\cG$ similarly to the one of $\cL$ :

\begin{equation}
    \begin{aligned}
    & p(ab|xy)  {\in} \cL   \;\; \Longleftrightarrow{}   \text{LP}_{\cL}(p(ab|xy) ) \;  \text{yields} \, v^* = 0 \\
    & p(ab|xy)  {\in} \cG   \;\; \Longleftrightarrow{}   \text{MILP}_{\cG}(p(ab|xy) ) \; \text{yields} \, c^* = 0
    \end{aligned}
\end{equation}
where $v^*$ is the minimum value $v \ge 0$ such that $(1-v) \, p(ab|xy) + v \,p_{L}(ab|xy) $ is in $\cL$, where $p_{L}(ab|xy)$ is a fixed behaviour in $\cL$. 

Thus, instead of an LP, one has to solve a MILP to determine whether a given behaviour is in $\cG$, which also admits tractable and efficient algorithms (although note that for both the membership in $\cL$ and $\cG$ one limitation is the number of local determinist strategies, increasing exponentially with the number of inputs). Finally, note that since we have proven that $\cG$ is a polytope, having its list of vertices would yield a similar LP-membership algorithm as the one for $\cL$. \footnote{While having its complete list of facets would provide a possibly even easier membership verification: check that the behaviour does not violate any of these facets. We say ``possibly" since the number of facets might be large.}

\emph{Discussion.---}  We have established that nonlocal games as originally defined, with a deterministic predicate, cannot, in general, be complete witnesses for nonlocality. Concretely, there are nonlocal behaviours that cannot provide a quantum advantage in any such nonlocal game. Furthermore, we showed that quantum mechanics can produce these nonlocal, game-classical  behaviours. On the road to answering this question we also establish interesting intermediate results and techniques. Namely, a Mixed-Integer-Linear-Program (MILP) that formalizes game equivalence for a Bell functional which allows us to obtain: i) the first examples of facet inequalities which are provably game-inequivalent; ii) hundreds of novel nonlocal game descriptions.  Furthermore, we established that the game-classical set of behaviours \(\cG\) is, in its own right, a polytope of which \(\cL\) is a proper subset of (except in some particular cases).  Since \(\cG\) is a polytope whose facets are game-equivalent, in a fixed scenario, it provides the smallest set of nonlocal games that enclose the local polytope. The relationship between \(\cG\) and \(\cQ\) still needs to be explored, namely, finding whether behaviours violating the facets exclusive to \(\cG\) can be quantum mechanical, and if these strictly game-nonlocal behaviours hold any operational meaning, other than being nonlocal. Finally, one could extend our analysis to scenarios where a nonlocal behaviour is used as a resource, together with free operations e.g., local operation with shared randomness (LOSR). We already know this is pertinent because all nonlocal behaviours in $(4,4,2,2)$ can be detected through input-lifting by nonlocal games in $(5,5,2,2)$, since we found a game-representative of the input-lifted version of $J_\star$. The question is open for scenario $(2,2,4,4)$, where no input-lifted game-representative of $F_\star^1$ and $F_\star^2$ can exist, and it might still be the case that nonlocal games are in fact LOSR-complete, i.e. that all nonlocal behaviours are game-activable. 

\emph{Data availability.---}
In our directory \cite{github_game_vs_inequ} we provide the algorithms and data discussed and used in this article. In particular, we provide the game representations of the 173 game-equivalent facet-inequalities in $(4,4,2,2)$, as well as the 31 game-equivalent facet-inequalities in $(2,2,4,4)$

\textit{AI disclosure.}
All analytical and numerical results were derived by the authors before the recent  advent of AI in the mathematics world. 
OpenAI's ChatGPT-6 Astra and Anthropic's Claude Opus 5.5 were used for text proofreading and producing the tikz code from a hand drawn sketch. The authors take full responsibility for the content of the work.

\emph{Acknowledgments.---}
The authors thank Sébastien Designolle for useful discussions on analytical certificates for LP and MILP, Emmanuel Zambrini Cruzeiro for providing the facets of scenarios $(4,4,2,2) $ and $(2,2,4,4)$, and Mateus Araújo for fruitful discussions.  RF acknowledges funding by FCT/MCTES - Fundação para a Ciência e a Tecnologia (Portugal). F.H. acknowledges funding by FCT/MCTES - Fundação para a Ciência e a Tecnologia (Portugal) through exploratory project 2023.12518.PEX.
This work was funded by QuantEdu France, a State aid managed by the French National Research Agency for France 2030 with the reference ANR-22-CMAS-0001 and the JCJC programme under grant number ANR-25-CE47-6396-01-HOQO-KS.

\appendix
\renewcommand*{\theHequation}{appendix.\thesection.\arabic{equation}}
\setcounter{secnumdepth}{1}


\section{Definitions and preliminary results} \label{app:definitions}

In this appendix we collect the main definitions and technical results concerning the characterization of Bell inequalities and nonlocal games. We restrict attention to bipartite scenarios.

First, we define the local polytope \(\mathcal{L}\). 

\begin{definition}[Local polytope]
\label{def:localpolytope}
Consider a bipartite Bell scenario \((X,Y,A,B)\), with conditional
probabilities \(p =\{p(ab|xy)\}_{abxy}\). Define the alphabets
\[
\mathsf X=\{0,\ldots,X-1\},\quad
\mathsf Y=\{0,\ldots,Y-1\},
\]
\[
\mathsf A=\{0,\ldots,A-1\},\quad
\mathsf B=\{0,\ldots,B-1\}.
\]
For deterministic assignments
\(f:\mathsf X\to\mathsf A\) and \(g:\mathsf Y\to\mathsf B\), define
\[
D^{f,g}_{ab|xy}
=
\delta_{a,f(x)}\delta_{b,g(y)} .
\]
The local polytope (V-description) is the convex hull of all deterministic local assignments:
\[
\mathcal L
=
\operatorname{conv}
\left\{
D^{f,g} : f:\mathsf X\to\mathsf A,\ g:\mathsf Y\to\mathsf B
\right\}.
\]

Equivalently, \(\mathcal L\)  is the set of behaviours that do not violate any
Bell inequality (H-description). That is,
\[
\mathcal{L}
= \{ p : 
\langle M,p\rangle \leq \beta_{\mathrm L}(M),
\quad
\forall M\},
\]
where
\[
\langle M,p\rangle
:=
\sum_{xyab} M_{ab|xy}\,p(ab|xy)
\]
and the local bound is
\[
\beta_{\mathrm L}(M)
=
\max_{f:\mathsf X\to\mathsf A,\;g:\mathsf Y\to\mathsf B}
\langle M,D^{f,g}\rangle .
\]
\end{definition}

Also, as defined in the main text (\cref{def:NLgame}), a Bell functional $G$ is a nonlocal game if there exists a probability distribution \(\mu(x,y)\) over the inputs, and a (deterministic) predicate $V(a,b|x,y)\in\{0,1\}$ 
such that \(
 G_{ab|xy}=\mu(x,y)V(a,b|x,y)\). 
Its associated local bound \(\nonumber \omega_{\mathrm L}(G)=\max_{p\in\cL}\ip{G}{p}\) is also called the classical value of the game. 

\begin{restatable}{lemma}{Bellone}
\label{prop:Bellgame}
A Bell functional \(M_{ab|xy}\) is equivalent to a nonlocal
game in the same scenario under block translation and positive
scaling if and only if, for every input pair \((x,y)\),
\[
M_{ab|xy}\in\{\lambda_1^{xy},\lambda_2^{xy}\}
\qquad\text{for all }a,b,
\]
where \(\lambda_1^{xy}\) and \(\lambda_2^{xy}\) are real
and distinct in nonconstant blocks. In this sense,
\[
\mu(x,y)V(a,b|x,y)\;\Leftrightarrow\;M_{ab|xy}.
\]
\end{restatable}

\begin{proof}
A nonlocal game has coefficients
\(G_{ab|xy}=\mu(x,y)V(a,b|x,y)\), so for each input
pair \((x,y)\) its entries belong to \(\{0,\mu(x,y)\}\).
Block translations and a common positive scaling preserve
the property of having at most two distinct values per block.
This proves one direction.

Conversely, suppose that
\(M_{ab|xy}\in\{\lambda_1^{xy},\lambda_2^{xy}\}\).
Without loss of generality, order the two levels as
\(\lambda_1^{xy}\leq\lambda_2^{xy}\), taking both equal
to the unique coefficient value for constant blocks.
By block translation,
\[
M'_{ab|xy}
:=M_{ab|xy}-\lambda_1^{xy}
\in\{0,\lambda_2^{xy}-\lambda_1^{xy}\}.
\]
Let
\[
S:=\sum_{x,y}
\bigl(\lambda_2^{xy}-\lambda_1^{xy}\bigr).
\]
If \(S>0\), define
\[
\mu(x,y)
:=\frac{\lambda_2^{xy}-\lambda_1^{xy}}{S},
\qquad
V(a,b|x,y):=\mathbf{1}[M'_{ab|xy}>0].
\]
Then \(\mu\) is a normalized input distribution,
\(V\) is a deterministic predicate, and
\[
M''_{ab|xy}:=\frac{M'_{ab|xy}}{S}
=\mu(x,y)V(a,b|x,y).
\]
Thus \(M\) is equivalent to a nonlocal game under
block translation and positive scaling.

If \(S=0\), all blocks are constant, and block
translation gives the zero game, represented by
\(V=0\) and any normalized input distribution.
\end{proof}

We show an important lemma characterizing \textit{no-signalling transformations}, which gives a useful parametrization of the transformations in \cref{def:Bellequiv}. 

\begin{lemma}
\label{lemma:NSbasis}
    Any \textit{strict NS}-equivalent representative of Bell functional \(M_{ab|xy}\) can be written in the form, 
\[
Z_{ab|xy}=M_{ab|xy}+r_{a|x}(y)+s_{b|y}(x),
\]
i.e \( T = r_{a|x}(y)+s_{b|y}(x)\) with \emph{zero-sum} constraints 
\begin{itemize}
    \item $\sum_{y}r_{a|x}(y)=0$ for each $(x,a)$ 
    \item $\sum_{x}s_{b|y}(x)=0$ for each $(y,b)$
\end{itemize}
\end{lemma}
\begin{proof}
Let \(Z_{ab|xy}=M_{ab|xy}+T_{ab|xy}\).
The transformation preserves the value of the functional on the no-signalling set iff
\begin{equation}
    \sum_{a,b,x,y}
    T_{ab|xy}\,
    p_{\rm NS}(ab|xy)
    =0
    \label{eq:T_NS_orthogonal}
\end{equation}
for every no-signalling behaviour \(p_{\rm NS}\). The form of \(T_{ab|xy}\) given in \cref{lemma:NSbasis} satisfies \cref{eq:T_NS_orthogonal} directly.

Conversely, suppose that \(T\) satisfies \cref{eq:T_NS_orthogonal}. Let \(\mathcal{NS}\) denote the linear subspace of arrays satisfying the no-signalling constraints. Then \(
    T\in \mathcal{NS}^{\perp} \), that is, \(T\) is the orthogonal space of no-signalling constraints.  Writing the no-signalling constraints as linear equations, \(T\) is given by linear combinations of the no-signalling constraint vectors.
 In bipartite scenarios, the two types of no-signalling constraints are
\[
\sum_a p_{\rm NS}(ab|xy)
    -
    \sum_a p_{\rm NS}(ab|x'y)
=0,
    \forall\,x,x',b,y,
    \]\[
    \sum_b p_{\rm NS}(ab|xy)
    -
    \sum_b p_{\rm NS}(ab|xy')
    =0,
    \forall\,y,y',a,x.
\]
The associated coefficient for the linear transformations are
\begin{subequations}
\begin{align}
    T^{(1)}_{ab|xy}
    &=
    \delta_{b b_0}\delta_{y y_0}
    \bigl(
        \delta_{x x_0}
        -
        \delta_{x x'_0}
    \bigr),
    \label{eq:T1_tensor}
    \\
    T^{(2)}_{ab|xy}
    &=
    \delta_{a a_0}\delta_{x x_0}
    \bigl(
        \delta_{y y_0}
        -
        \delta_{y y'_0}
    \bigr).
    \label{eq:T2_tensor}
\end{align}
\end{subequations}
The first type is independent of \(a\), whereas the second type is independent of \(b\). Moreover,
\begin{subequations}
\begin{align}
    \sum_x T^{(1)}_{ab|xy}
    &=
    \delta_{b b_0}\delta_{y y_0}
    \sum_x
    \bigl(
        \delta_{x x_0}
        -
        \delta_{x x'_0}
    \bigr)
    =0,
    \label{eq:T1_zero_sum}
    \\
    \sum_y T^{(2)}_{ab|xy}
    &=
    \delta_{a a_0}\delta_{x x_0}
    \sum_y
    \bigl(
        \delta_{y y_0}
        -
        \delta_{y y'_0}
    \bigr)
    =0.
    \label{eq:T2_zero_sum}
\end{align}
\end{subequations}

Thus, a generic tensor \(T\in\mathcal{NS}^{\perp}\) can be written as
\begin{equation}
\begin{split}
    T_{ab|xy}
    &=
    \sum_{b_0,y_0,x_0,x'_0}
    \lambda^{(1)}_{b_0 y_0 x_0 x'_0}
    T^{(1)}_{ab|xy}
    \\
    &\quad+
    \sum_{a_0,x_0,y_0,y'_0}
    \lambda^{(2)}_{a_0 x_0 y_0 y'_0}
    T^{(2)}_{ab|xy}.
\end{split}
\label{eq:T_linear_combination}
\end{equation}
Since the first sum is independent of \(a\), and the second is independent of \(b\), define
\begin{subequations}
\begin{align}
    s_{b|y}(x)
    &:=
    \sum_{b_0,y_0,x_0,x'_0}
    \lambda^{(1)}_{b_0 y_0 x_0 x'_0}
    T^{(1)}_{ab|xy},
    \label{eq:s_definition}
    \\
    r_{a|x}(y)
    &:=
    \sum_{a_0,x_0,y_0,y'_0}
    \lambda^{(2)}_{a_0 x_0 y_0 y'_0}
    T^{(2)}_{ab|xy}.
    \label{eq:r_definition}
\end{align}
\end{subequations}
Therefore
\begin{equation}
    T_{ab|xy}
    =
    r_{a|x}(y)+s_{b|y}(x).
    \label{eq:T_rs_decomposition}
\end{equation}
Furthermore, since every term in the first sum satisfies \cref{eq:T1_zero_sum}, and every term in the second sum satisfies \cref{eq:T2_zero_sum}, one has
\begin{subequations}
\begin{align}
    \sum_x s_{b|y}(x)
    &=0,
    &
    \forall\,b,y,
    \label{eq:s_zero_sum}
    \\
    \sum_y r_{a|x}(y)
    &=0,
    &
    \forall\,a,x.
    \label{eq:r_zero_sum}
\end{align}
\end{subequations}
Thus, every tensor orthogonal to the no-signalling subspace can be parametrized, without loss of generality, as\(
    T_{ab|xy}
    =
    r_{a|x}(y)+s_{b|y}(x),
\)
with
\begin{equation}
    \sum_y r_{a|x}(y)=0,
    \qquad
    \sum_x s_{b|y}(x)=0.
\end{equation}
\end{proof}

Including positive scaling and block translations, the general
equivalent representative takes the form
\begin{equation}
\label{eq:equivclass}
 Z_{ab|xy}
=cM_{ab|xy}+r_{a|x}(y)+s_{b|y}(x)+d_{xy},   
\end{equation}
where \(c>0\), \(d_{xy}\in\mathbb R\), and \(r,s\) satisfy
the zero-sum constraints above. Thus, \cref{eq:equivclass} gives a  practical parametrization of  \cref{def:Bellequiv} for game-equivalence, with \(\Gamma_{ab|xy} = r_{a|x}(y)+s_{b|y}(x)+d_{xy}\).

\medskip

\section{Formalizing game equivalence with Mixed-Integer-Linear-Programming}
\label{app:MILP}
Using \cref{eq:equivclass}, one can reframe the question of  game-equivalence of a Bell functional as a One-Shot Mixed-Integer-Linear-Program (MILP). The program is specified below:
\begin{widetext}
\begin{algorithm}[H]
\caption{One-shot MILP test for game equivalence of a Bell functional}
\label{alg:oneshot_milp}

\KwIn{Bell functional \(M=\{M_{ab|xy}\}_{a,b,x,y}\), input sets \(x,y\), output sets \(a,b\).}
\KwOut{A predicate \(V\) and prior \(\mu\), or a statement of game inequivalence.}

\BlankLine
Introduce binary variables
\[
    V_{ab|xy}\in\{0,1\},
    \qquad \forall\,a,b,x,y .
\]

Introduce continuous variables
\[
    Z_{ab|xy}\in\mathbb R,
    \qquad
    \lambda_{xy}\in[0,1],
    \qquad
    d_{xy}\in\mathbb R,
    \qquad
    c\ge 0 .
\]

Introduce no-signalling shift variables
\[
    r_{a|x}(y)\in\mathbb R,
    \qquad
    s_{b|y}(x)\in\mathbb R .
\]

\BlankLine
Solve the following MILP:
\[
\begin{aligned}
c^\star=\max\quad & c \\[1mm]
\mathrm{s.t.}\quad
&
Z_{ab|xy}
=
c\,M_{ab|xy}
+
r_{a|x}(y)
+
s_{b|y}(x)
+
d_{xy},
&& \forall\,a,b,x,y,
\\[1mm]
&
\sum_y r_{a|x}(y)=0,
&& \forall\,a,x,
\\[1mm]
&
\sum_x s_{b|y}(x)=0,
&& \forall\,b,y,
\\[1mm]
&
0\le Z_{ab|xy}\le V_{ab|xy},
&& \forall\,a,b,x,y,
\\[1mm]
&
\lambda_{xy}-(1-V_{ab|xy})
\le
Z_{ab|xy}
\le
\lambda_{xy}+(1-V_{ab|xy}),
&& \forall\,a,b,x,y,
\\[1mm]
&
V_{ab|xy}\in\{0,1\},
&& \forall\,a,b,x,y,
\\[1mm]
&
0\le \lambda_{xy}\le 1,
&& \forall\,x,y,
\\[1mm]
&
c\ge 0 .
\end{aligned}
\]

\BlankLine
\If{\(c^\star = 0\)}{
    \Return{Bell functional is game-inequivalent}
}
\Else{
    Extract the optimal values of \(V_{ab|xy}\) and \(\lambda_{xy}\)\;

    Define
    \[
        \mu(x,y)
        =
        \frac{\lambda_{xy}}
        {\sum_{x',y'}\lambda_{x'y'}} .
    \]

    \Return{\(V,\mu\)}
}

\end{algorithm}
\end{widetext}

The scalar \(c\) accounts for the irrelevant positive rescaling
of the Bell functional. For a functional \(M\) that is nonconstant
on the no-signalling set, an exact feasible solution with \(c>0\)
certifies that \(M\) is no-signalling equivalent, up to block shifts,
to a nonlocal game. The constraints force each adjusted block
\(Z_{\cdot\cdot|xy}\) to take only the two values \(0\) and
\(\lambda_{xy}\): if \(V_{ab|xy}=0\), then \(Z_{ab|xy}=0\),
while if \(V_{ab|xy}=1\), then \(Z_{ab|xy}=\lambda_{xy}\).
The normalized block levels define the prior \(\mu\).

If the exact optimum is \(c^\star=0\), there is no solution
with positive scaling, so \(M\) is not game equivalent.
Here we assume that \(M\) takes different values on at least
two no-signalling behaviours. If instead it takes the same
value on every no-signalling behaviour, it is already
equivalent to a trivial game and no MILP test is needed.

In practice, numerical implementations test for
\(c^\star>\varepsilon\), within a fixed solver tolerance
\(\varepsilon>0\). A returned representation provides a
rigorous certificate of game equivalence when its predicate,
prior, and transformation identities are verified exactly.
Failure to find a solution with \(c>\varepsilon\), however,
does not by itself provide a definite proof of game
inequivalence.

\section{\(J_\star\) is game-inequivalent in \((4,4,2,2)\)}
\label{app:J40}

In this section we prove that \(J_\star\) is the only
game-inequivalent facet inequality in the scenario \((4,4,2,2)\), following
the numerical evidence obtained from the MILP in
Appendix~\ref{app:MILP}. We restate \cref{lemma:4422} for convenience.

\lemmaFourFourTwoTwo*

First, we will prove a small combinatorial result which will be useful in what follows. 

\begin{remark}
\label{lemma:combinatorial}
   Let \(\lambda\in\mathbb R\) and \(S_\lambda:=\{-\lambda,0,\lambda\}\). If \(r\in S_\lambda\) and \(r+3\in S_\lambda\), then necessarily
\[
r\in \{-3, -3/2, 0\}.
\] 
\end{remark}
\begin{proof}
If \(r \in S_\lambda\) then either
    \begin{enumerate}
  \item \(r=0 \Rightarrow 3 \in S_\lambda \Rightarrow  \lambda = \pm3\).
  \item  \(r=\lambda \Rightarrow 3+\lambda \in S_\lambda \) thus,
  \begin{itemize}
      \item \(3+\lambda = - \lambda \Leftrightarrow \lambda = -3/2\)
      \item \(3+\lambda = 0 \Leftrightarrow \lambda = -3  \)
       \item \(3+\lambda =  \lambda \Leftrightarrow 3=0\)
  \end{itemize}
  Since the last one is impossible, then either \(r=-3\) or \(r=-3/2\).
  \item  \(r=-\lambda \Rightarrow 3-\lambda \in S_\lambda \) thus,
  \begin{itemize}
      \item \(3-\lambda = - \lambda \Leftrightarrow 3 = 0\)
      \item \(3-\lambda = 0 \Leftrightarrow \lambda = 3  \)
       \item \(3-\lambda =  \lambda \Leftrightarrow \lambda = 3/2\)
  \end{itemize}
  Since the first one is impossible, then either \(r=-3\) or \(r=-3/2\).
\end{enumerate}
\end{proof}
Let us give a representation of the Bell functional in full-probability notation. 
For each $(x,y)$ we obtain a \(2\times2\) block matrix
\[
M_{ab|xy}
=
\left(
\begin{array}{cc}
M_{xy}^{00} & M_{xy}^{01}\\
M_{xy}^{10} & M_{xy}^{11}
\end{array}
\right),
\]
The entire functional \({J_\star}_{ab|xy}\) is given below. Rows are grouped by Alice's input and columns by Bob's input; each block is indexed by the outputs:
\begin{equation}
    {J_\star}_{ab|xy} = 
\left(
\begin{array}{cc|cc|cc|cc}
-7 & -5 &  3 & 0 &  3 & 0 &  0 & 0\\
-5 &  0 &  0 & 0 &  0 & 0 & -3 & 0\\ \hline
 4 &  1 & -1 & 0 & -1 & 0 & -2 & 0\\
 0 &  0 &  0 & 0 &  0 & 0 &  0 & 0\\ \hline
 1 &  0 &  1 & 0 & -3 & 0 &  2 & 0\\
 0 &  0 &  0 & 0 &  0 & 0 &  0 & 0\\ \hline
 1 &  0 & -3 & 0 &  1 & 0 &  2 & 0\\
 0 &  0 &  0 & 0 &  0 & 0 &  0 & 0
\end{array}
\right)
\end{equation}

Now let us prove game inequivalence. 

\begin{proof}

From \cref{prop:Bellgame} and \cref{lemma:NSbasis} we know that if the functional admitted a description as a nonlocal game, then
\[
Z_{ab|xy}
=
{J_\star}_{ab|xy}
+
r_{a|x}(y)
+
s_{b|y}(x)
\in
\{\lambda^{xy}_1,\lambda^{xy}_2\},
\]
with \( \sum_y r_{a|x}(y)=0
\) for each \((x,a)\) and 
\(
\sum_x s_{b|y}(x)=0
\), for each \((y,b)\).

Indeed, for a nonlocal game each block must take at most two values
\(\{\lambda^{xy}_1,\lambda^{xy}_2\}\), hence \emph{every} within-block row/column difference belongs to
\[
\{-\Delta_{xy},0,\Delta_{xy}\},
\qquad
\Delta_{xy}:=\lambda^{xy}_2-\lambda^{xy}_1 .
\]
Focusing on the first row of four \(2\times2\) matrices, with \(x=0\) fixed, let us compute column differences
\[
Z_{ab|0y}
=
{J_\star}_{ab|0y}
+
r_{a|0}(y)
+
s_{b|y}(0)
\]
as a way to eliminate the column \(s\)-variables.

The left column differences
\[
\operatorname{lcd}(y)
=
Z_{00|0y}
-
Z_{10|0y},
\qquad
y\in\{0,1,2,3\},
\]
have the general structure
\[
\begin{aligned}
\operatorname{lcd}(y)
&=
\left({J_\star}_{00|0y}-{J_\star}_{10|0y}\right)
+
\left(r_{0|0}(y)-r_{1|0}(y)\right)
+\\
&
\left(s_{0|y}(0)-s_{0|y}(0)\right)
=
\left({J_\star}_{00|0y}-{J_\star}_{10|0y}\right)
+
\rho_y ,
\end{aligned}
\]
where
\[
\rho_y
=
r_{0|0}(y)-r_{1|0}(y).
\]
Likewise, the right column differences
\[
\operatorname{rcd}(y)
=
Z_{01|0y}
-
Z_{11|0y},
\qquad
y\in\{0,1,2,3\},
\]
are given by
\[
\begin{aligned}
\operatorname{rcd}(y)
&=
\left({J_\star}_{01|0y}-{J_\star}_{11|0y}\right)
+
\left(r_{0|0}(y)-r_{1|0}(y)\right)
+\\
&
\left(s_{1|y}(0)-s_{1|y}(0)\right)
=
\left({J_\star}_{01|0y}-{J_\star}_{11|0y}\right)
+
\rho_y .
\end{aligned}
\]

So for all \(y\) we have explicitly
\[
\begin{aligned}
\left\{
\begin{array}{l}
\operatorname{lcd}(0)=-2+\rho_0,\\[2pt]
\operatorname{rcd}(0)=-5+\rho_0,
\end{array}
\right.
&\qquad
\left\{
\begin{array}{l}
\operatorname{lcd}(1)=3+\rho_1,\\[2pt]
\operatorname{rcd}(1)=\rho_1,
\end{array}
\right.
\\[8pt]
\left\{
\begin{array}{l}
\operatorname{lcd}(2)=3+\rho_2,\\[2pt]
\operatorname{rcd}(2)=\rho_2,
\end{array}
\right.
&\qquad
\left\{
\begin{array}{l}
\operatorname{lcd}(3)=3+\rho_3,\\[2pt]
\operatorname{rcd}(3)=\rho_3.
\end{array}
\right.
\end{aligned}
\]

Two-valuedness of each block therefore demands that
\[
\begin{aligned}
\operatorname{lcd}(0),\operatorname{rcd}(0)
&\in
\{-\Delta_{x0},0,\Delta_{x0}\},
\\
\operatorname{lcd}(1),\operatorname{rcd}(1)
&\in
\{-\Delta_{x1},0,\Delta_{x1}\},
\\
\operatorname{lcd}(2),\operatorname{rcd}(2)
&\in
\{-\Delta_{x2},0,\Delta_{x2}\},
\\
\operatorname{lcd}(3),\operatorname{rcd}(3)
&\in
\{-\Delta_{x3},0,\Delta_{x3}\}.
\end{aligned}
\]

Invoking now  \cref{lemma:combinatorial}, we have that for \(y=1,2,3\), the adherence of the column differences to two-block values demands
\[
\rho_y\in\{-3,-3/2,0\},
\qquad
y=1,2,3.
\tag{A}
\]
For \(y=0\), we define
\[
r=-5+\rho_0,
\qquad
r+3=-2+\rho_0,
\]
and apply the lemma to obtain
\[
-5+\rho_0\in\{-3,-3/2,0\},
\]
which in turn implies
\[
\rho_0\in\{2,7/2,5\}.
\tag{B}
\]

To finalize the proof, let us show that the transformations that allow the column differences to adhere to this two-value constraint are supposed to be NS. As such, they would need to obey
\[
\rho_0+\rho_1+\rho_2+\rho_3=0,
\]
or equivalently,
\[
\rho_0
=
-\left(\rho_1+\rho_2+\rho_3\right).
\]
But since we have that
\[
\rho_0\in\{2,\frac72,5\}
\]
and
\[
-\left(\rho_1+\rho_2+\rho_3\right)
\in
\{0,\frac32,3,\frac92,6,\frac{15}{2},9\},
\]
which have empty intersection, the zero-sum constraint cannot be satisfied.

This contradicts the assumption that all four blocks can be made to have at most two values under a \emph{valid} NS transformation in 4422. Therefore no strict NS-equivalent representative of the facet can be a deterministic-predicate game in 4422.

Furthermore, since we have explicit nonlocal game descriptions for all other facets in the scenario (see \cite{github_game_vs_inequ}) we have proven \cref{lemma:4422}.
\end{proof}

\section{\(F_\star^{1}\) and \(F_\star^{2}\) remain game-inequivalent under standard liftings}
\label{app:F26F29}

In Araújo et al.~\cite{araujo2020bell}, it was shown that game-equivalence is guaranteed for binary output inequalities by allowing larger input scenarios. We restate the theorem verbatim from Appendix~B.

\begin{unnumberedtheorem}[Araújo et al., Appendix B]
Let $M$ be a bipartite Bell functional with $s$ inputs and two outputs per party.
There exists a deterministic nonlocal game $G$ with $2s$ inputs and two outputs
per party which is equivalent to $M$.
\end{unnumberedtheorem}

Ideally we would like to have a proof of game-inequivalence that would be immune to liftings to larger input and output scenarios. A sufficient criterion to establish this was also presented in Appendix A of Araújo et al.~\cite{araujo2020bell}. The proof therein shows inequivalence just using a single input block of the functional, so it is completely agnostic to the number of inputs of the scenario. We will use this proof technique to prove \cref{lemma:2244}, which we restate here for convenience. 

\lemmaTwoTwoFourFour*

The two facets \(F_\star^{1}\) and \(F_\star^{2}\), represented below in full-probability notation. Rows are grouped by Alice's input and columns by Bob's input; each \(4\times4\) block is indexed by the outputs. For \(F_\star^{1}\) we have
\begin{equation}
    {F_\star^{1}}_{ab|xy}
=
\left(
\begin{array}{cccc|cccc}
-3 & -3 & -3 & -3 &  3 &  2 &  1 & 0\\
-3 & -2 & -2 & -2 &  2 &  1 &  1 & 0\\
-3 & -2 & -1 & -1 &  1 &  1 &  0 & 0\\
-3 & -2 & -1 &  0 &  0 &  0 &  0 & 0\\ \hline
 3 &  2 &  1 &  0 & -3 & -2 & -1 & 0\\
 2 &  1 &  1 &  0 & -2 & -1 &  0 & 0\\
 1 &  1 &  0 &  0 & -1 &  0 &  0 & 0\\
 0 &  0 &  0 &  0 &  0 &  0 &  0 & 0
\end{array}
\right).
\end{equation}

and for \(F_\star^{2}\),
\begin{equation}
    {F_\star^{2}}_{ab|xy} =
\left(
\begin{array}{cccc|cccc}
-3 & -3 & -3 & -3 &  3 &  2 &  1 & 0\\
-3 & -2 & -3 & -2 &  2 &  2 &  1 & 0\\
-3 & -3 & -4 & -2 &  2 &  2 &  2 & 0\\
-3 & -2 & -2 &  0 &  0 &  0 &  0 & 0\\ \hline
 3 &  2 &  2 &  0 & -3 & -2 & -2 & 0\\
 2 &  2 &  2 &  0 & -2 & -2 & -1 & 0\\
 1 &  1 &  2 &  0 & -2 & -1 &  0 & 0\\
 0 &  0 &  0 &  0 &  0 &  0 &  0 & 0
\end{array}
\right).
\end{equation}

A proof analogous to that for \(J_\star\) would establish the following weaker statement for \(F_\star^{1}\) and \(F_\star^{2}\).

\medskip

\noindent\textbf{Weak statement.}
\(F_\star^{1}\) and \(F_\star^{2}\) do not admit a description as nonlocal games in scenario $(2,2,4,4)$.

We will actually prove a stronger statement, from which the weak statement follows immediately.

\medskip

\noindent\textbf{Strong statement.}
Any bipartite lifted version of \(F_\star^{1}\) and \(F_\star^{2}\) do not admit a description as nonlocal games in scenario $(X,Y,A,B)$, with
$X,Y\geq2$ and $A,B\geq4$.

\medskip

We prove this by showing that, for each inequality, a single input block cannot be made two-valued by any choice of row and column offsets. Thus the contradiction holds for any lifting of inputs and outputs. 

\begin{proof}
Consider the first block of \(F_\star^{1}\),
\begin{equation*}
\left(
\begin{array}{cccc}
-3 & -3 & -3 & -3\\
-3 & -2 & -2 & -2\\
-3 & -2 & -1 & -1\\
-3 & -2 & -1 & 0
\end{array}
\right).
\end{equation*}
We can add a global shift of \(+3\) for free. From NS, add column offsets \(b,b',b''\) to columns \(2,3,4\), respectively, and row offsets \(a,a',a''\) to rows \(2,3,4\), respectively. The resulting matrix is

\[
M=
\begin{pmatrix}
0 & b & b' & b''\\
a & 1+b+a & 1+b'+a & 1+b''+a\\
a' & 1+b+a' & 2+b'+a' & 2+b''+a'\\
a'' & 1+b+a'' & 2+b'+a'' & 3+b''+a''
\end{pmatrix}
\]

Here we allow arbitrary row and column offsets within the selected block, without imposing the zero-sum constraints on the no-signalling shifts. This enlarges the allowed class of transformations: impossibility even with this freedom therefore rules out a no-signalling transformation making the block two-valued. By definition, all entries of the block must take at most two distinct values, namely \(0\), because we have already shifted globally the entire block, and some common nonzero value \(\lambda\). From the first row and first column it follows immediately that
\begin{equation*}
a,a',a'',b,b',b''\in\{0,\lambda\}.
\end{equation*}

Consider \(M_{34}=2+b''+a'\). Since \(a',b''\in\{0,\lambda\}\), the possible values of \(M_{34}\) are
\begin{equation*}
2,\qquad 2+\lambda,\qquad 2+\lambda,\qquad 2+2\lambda.
\end{equation*}
For \(M_{34}\) to belong to \(\{0,\lambda\}\), at least one of these values must equal either \(0\) or \(\lambda\). This forces
\begin{equation*}
\lambda\in\{2,-2,-1\}.
\end{equation*}

Next consider \(M_{44}=3+b''+a''\). Again using \(a'',b''\in\{0,\lambda\}\), the possible values of \(M_{44}\) are
\begin{equation*}
3,\qquad 3+\lambda,\qquad 3+\lambda,\qquad 3+2\lambda.
\end{equation*}
But since \(\lambda\in\{2,-2,-1\}\), we have
\begin{align*}
\lambda=2
&\quad\Longrightarrow\quad
M_{44}\in\{3,5,7\},\\
\lambda=-2
&\quad\Longrightarrow\quad
M_{44}\in\{3,1,-1\},\\
\lambda=-1
&\quad\Longrightarrow\quad
M_{44}\in\{3,2,1\}.
\end{align*}
In none of these cases can \(M_{44}\) belong to \(\{0,\lambda\}\). Therefore, no choice of row and column offsets can make this block take only two values, \(0\) and a common nonzero value \(\lambda\).

Now consider the first block of \(F_\star^{2}\),
\begin{equation*}
\left(
\begin{array}{cccc}
-3 & -3 & -3 & -3\\
-3 & -2 & -3 & -2\\
-3 & -3 & -4 & -2\\
-3 & -2 & -2 & 0
\end{array}
\right).
\end{equation*}
After applying the same global shift by \(+3\), and the same row and column offsets, we obtain

\begin{equation*}
M
=
\left(
\begin{array}{cccc}
0 & b & b' & b''\\
a & 1+b+a & b'+a & 1+b''+a\\
a' & b+a' & -1+b'+a' & 1+b''+a'\\
a'' & 1+b+a'' & 1+b'+a'' & 3+b''+a''
\end{array}
\right)
\end{equation*}

Again, the first row and first column imply
\begin{equation*}
a,a',a'',b,b',b''\in\{0,\lambda\}.
\end{equation*}

Focus on
\begin{equation*}
M_{33}=-1+a'+b'.
\end{equation*}
Since \(a',b'\in\{0,\lambda\}\), the possible values of \(M_{33}\) are
\begin{equation*}
-1,\qquad -1+\lambda,\qquad -1+\lambda,\qquad -1+2\lambda.
\end{equation*}
For \(M_{33}\) to belong to \(\{0,\lambda\}\), at least one of these values must equal either \(0\) or \(\lambda\). This forces
\begin{equation*}
\lambda\in\left\{-1,1,\frac12\right\}.
\end{equation*}

However, \(M_{44}=3+a''+b''\) can only take the values
\begin{equation*}
3,\qquad 3+\lambda,\qquad 3+\lambda,\qquad 3+2\lambda.
\end{equation*}
For the three possible values of \(\lambda\), this gives
\begin{align*}
\lambda=-1
&\quad\Longrightarrow\quad
M_{44}\in\{3,2,1\},\\
\lambda=1
&\quad\Longrightarrow\quad
M_{44}\in\{3,4,5\},\\
\lambda=\frac12
&\quad\Longrightarrow\quad
M_{44}\in\left\{3,\frac72,4\right\}.
\end{align*}
In none of these cases can \(M_{44}\) belong to \(\{0,\lambda\}\). This gives a contradiction for \(F_\star^{2}\).

Hence both \(F_\star^{1}\) and \(F_\star^{2}\) admit a single-input-block obstruction. Therefore both remain game-inequivalent under every standard input or output lifting.

Furthermore, since we have explicit nonlocal game descriptions for all other facets in the scenario (see \cite{github_game_vs_inequ}) we have proven \cref{lemma:2244}.
\end{proof}

Thus, though providing stronger statements, the proofs of \(F_\star^{1}\) and \(F_\star^{2}\) game-inequivalences were easier to establish, since the contradiction already arises within a single input block and does not require the zero-sum constraints used for \(J_\star\). 

\section{$\mathcal{G}$ is a polytope with game facet-inequalities} \label{app:G is a polytope with game facets}

Here we prove that the set of behaviours that do not violate any nonlocal game forms a polytope. Moreover, we will see that the facets of that polytope are all nonlocal games.

The set in question $\cG$ is defined as
\begin{equation}
  \cG = \{\, p \in \NS \mid G(p) \leq \omega_c(G) \;\; \forall G \,\}
\end{equation}
where $\NS $ denotes the set of valid no-signaling behaviours, $G(p)$ is the
value of a game applied on $p$, and $\omega_c(G)$ the classical value
of the game.

Let us focus on the bipartite case for convenience (while keeping in
mind that all steps apply to generic multipartite behaviours) and write
the ``non-game violation'' condition explicitly:
\begin{equation} \label{eq:no-violation}
  \begin{gathered}
    G\bigl(p(ab|xy)\bigr) \leq \omega_c(G) \iff \\
    \sum_{abxy} \mu(x,y) V(a,b|x,y)\, p(ab|xy) \\
    \leq \max_s \sum_{abxy} \mu(x,y) V(a,b|x,y)\, D^{(s)}(ab|xy)
  \end{gathered}
\end{equation}
where $D^{(s)}(ab|xy)$ is one of the $N$ deterministic strategies,
indexed by an integer $s \in \{1,\dots,N\}$.

For a behaviour $p(ab|xy)$ and a fixed predicate $V$, let us define
\begin{equation}
  S_V(p)(x,y) \equiv \sum_{ab} V(a,b|x,y)\, p(ab|xy)
\end{equation}
that is, $S_V(p)(x,y)$ is the score of behaviour $p$ for predicate $V$
on the pair of inputs $x,y$; and likewise the score of the
deterministic strategy $s$,
\begin{equation}
  g_V^{(s)}(x,y) \equiv \sum_{ab} V(a,b|x,y)\, D^{(s)}(ab|xy) .
\end{equation}

\paragraph*{Notation.}
From now on we regard $S_V(p)$, $g_V^{(s)}$ and $\mu$ as vectors of
$\R^n$ indexed by the $n$ input pairs $(x,y)$, we write
$\ip{\cdot}{\cdot}$ for the Euclidean product,
$\ip{\mu}{S_V(p)} = \sum_{xy} \mu(x,y) S_V(p)(x,y)$, and we denote by
$\Delx$ and $\Delq$ the probability simplices over the input pairs and
over the $N$ deterministic strategies respectively (that is, vectors of ($X\cdot Y$, respectively $N$) nonnegative components which sum to $1$). Inequalities
between vectors are understood componentwise.

With this notation, the set $\cG_V$ of behaviours that do not violate
any game built on the fixed predicate $V$ (with arbitrary prior) reads
\begin{equation} \label{eq:GV}
  \cG_V = \bigl\{\, p \in \NS \;\bigm|\;
  \ip{\mu}{S_V(p)} \leq \max_s \ip{\mu}{g_V^{(s)}}
  \;\; \forall \mu \in \Delx \,\bigr\} .
\end{equation}

The main step of the proof is the following lemma, which characterizes
the space defined by condition~\eqref{eq:GV}.

\begin{lemma} \label{D_V lemma}
  $p \in \cG_V$ if and only if $S_V(p) \in D_V$, where
  $D_V = \conv\bigl(\{g_V^{(s)}\}_s\bigr) - \R^n_{\geq 0}$.
\end{lemma}

Explicitly,
\begin{equation} \label{eq:DV}
  D_V = \Bigl\{\, d \in \R^n \;\Bigm|\;
  d = \sum_k q_k\, g_V^{(k)} - u , \;\;
  q \in \Delq , \; u \geq 0 \,\Bigr\} ,
\end{equation}
i.e.\ the convex hull of the $g_V^{(k)}$ shifted by arbitrary
nonpositive vectors.

\begin{proof}[Proof of \cref{D_V lemma}, ``if'' direction]
  If $S_V(p) \in D_V$ there exist $q \in \Delq$ and $u \geq 0$ with
  \begin{equation}
    S_V(p) = \sum_k q_k\, g_V^{(k)} - u .
  \end{equation}
  Evaluating against an arbitrary $\mu \in \Delx$,
  \begin{align}
    \ip{\mu}{S_V(p)}
      &= \sum_k q_k \ip{\mu}{g_V^{(k)}} - \ip{\mu}{u} \notag \\
      &\leq \max_k \ip{\mu}{g_V^{(k)}} ,
  \end{align}
  since $\ip{\mu}{u} \geq 0$ and a convex combination is bounded by its
  largest term. Hence $p \in \cG_V$.
\end{proof}

For the converse we use the following intermediate lemma.

\begin{lemma} \label{lem:minimax}
  One has $\ip{\mu}{S_V(p)} \leq \max_s \ip{\mu}{g_V^{(s)}}$ for all
  $\mu \in \Delx$ if and only if there exists $q \in \Delq$ such that
  $S_V(p)(x,y) \leq \sum_k q_k\, g_V^{(k)}(x,y)$ for all $x,y$.
\end{lemma}

\begin{proof}[Proof of \cref{D_V lemma}, ``only if'' direction]
  Let $p \in \cG_V$. By \cref{lem:minimax} there is a
  $q \in \Delq$ with
  $S_V(p)(x,y) \leq \sum_k q_k\, g_V^{(k)}(x,y)$ for all $x,y$.
  Setting
  \begin{equation}
    u \equiv \sum_k q_k\, g_V^{(k)} - S_V(p) \geq 0 ,
  \end{equation}
  we get $S_V(p) = \sum_k q_k\, g_V^{(k)} - u$, i.e.\
  $S_V(p) \in D_V$.
\end{proof}

\begin{proof}[Proof of \cref{lem:minimax}]
  Fix $\mu \in \Delx$. Then
  \begin{align}
    &\ip{\mu}{S_V(p)} \leq \max_s \ip{\mu}{g_V^{(s)}} \notag \\
    \iff\; &\min_s \ip{\mu}{S_V(p) - g_V^{(s)}} \leq 0 \notag \\
    \iff\; &\min_{q \in \Delq}
             \Bigl\langle \mu , S_V(p) - \textstyle\sum_k q_k g_V^{(k)}
             \Bigr\rangle \leq 0 ,
  \end{align}
  where in the last step we used that the minimum of a linear function
  over a polytope is attained at a vertex, so that minimizing over the
  whole simplex is equivalent to minimizing over its vertices.

  Introduce the residual vector
  \begin{equation}
    h_q \equiv S_V(p) - \sum_k q_k\, g_V^{(k)} \in \R^n .
  \end{equation}
  Requiring the condition for \emph{all} $\mu \in \Delx$ amounts to
  \begin{equation}
    \max_{\mu \in \Delx} \; \min_{q \in \Delq} \; \ip{\mu}{h_q} \leq 0 .
  \end{equation}
  The function $(\mu,q) \mapsto \ip{\mu}{h_q}$ is linear in $\mu$ and
  affine in $q$, and both are optimized over compact convex sets, so
  the minimax theorem allows us to swap the two optimizations \cite{Simons1995_minimax}.
  Moreover, maximizing a convex combination of the entries of $h_q$
  simply selects the largest entry. Therefore
  \begin{align}
    &\max_{\mu \in \Delx} \min_{q \in \Delq} \ip{\mu}{h_q} \leq 0
      \notag \\
    \iff\; &\min_{q \in \Delq} \max_{\mu \in \Delx} \ip{\mu}{h_q}
      \leq 0 \notag \\
    \iff\; &\min_{q \in \Delq} \; \max_{x,y} \; h_q(x,y) \leq 0
      \notag \\
    \iff\; &\exists\, q \in \Delq : \; h_q(x,y) \leq 0
      \;\; \forall x,y ,
  \end{align}
  which is the claim.
\end{proof}

\subsection{$\cG$ has only nonlocal game-equivalent facet-inequalities}

\cref{D_V lemma} replaced the infinitely many
inequalities~\eqref{eq:GV} by the single membership condition
$S_V(p) \in D_V$. The set $D_V$ of \cref{eq:DV} is an unbounded
polyhedron: the polytope
$P_g = \conv\bigl(\{g_V^{(k)}\}_k\bigr)$ extended towards $-\infty$ in
all directions. Such a set is characterized by finitely many
facet-inequalities with nonnegative coefficients,
\begin{equation}
  d \in D_V \iff \ip{c_k}{d} \leq b_k \;\; \forall k ,
  \qquad c_k \geq 0 ,
\end{equation}
as formalized in \cref{lem:polyhedron} below.

We thus have a finite list of inequalities with nonnegative
coefficients characterizing $S_V(p) \in D_V$, namely
$\ip{c_k}{S_V(p)} \leq b_k$. Such inequalities are defined up to a
positive rescaling, so we may normalize them,
\begin{equation}
  \frac{1}{\sum_{xy} c_k(x,y)} \ip{c_k}{S_V(p)} \leq \tilde{b}_k ,
\end{equation}
and define $\mu_k(x,y) \equiv c_k(x,y) / \sum_{x'y'} c_k(x',y')$, so
that $\mu_k \in \Delx$. The corresponding bound $\tilde b_k$ is the
maximal value of the left-hand side over $D_V$, which -- all
coefficients being nonnegative -- is attained on $P_g$:
\begin{equation}
  \tilde{b}_k = \max_s \ip{\mu_k}{g_V^{(s)}} .
\end{equation}
Recalling that $g_V^{(s)}(x,y) = \sum_{ab} V(a,b|x,y) D^{(s)}(ab|xy)$
we can now conclude.

\begin{tcolorbox}
  For a fixed predicate $V(a,b|x,y)$, a behaviour is in $\cG_V$ iff
  $S_V(p) \in D_V$, that is, iff for all $k$
  \begin{equation*}
    \ip{\mu_k}{S_V(p)} \leq \max_s \ip{\mu_k}{g_V^{(s)}} ,
  \end{equation*}
  or equivalently
  \begin{equation*}
    \begin{gathered}
      \sum_{abxy} \mu_k(x,y) V(a,b|x,y)\, p(ab|xy) \\
      \leq \max_s \sum_{abxy} \mu_k(x,y) V(a,b|x,y)\, D^{(s)}(ab|xy) ,
    \end{gathered}
  \end{equation*}
  where the $\{\mu_k\}$ are the normalized coefficients of the facets
  of $D_V$.

  Since there are finitely many predicates,
  $\cG = \bigcap_V \cG_V$ is a finite intersection of polytopes, hence
  a polytope.
\end{tcolorbox}

\begin{remark}
  We have finitely many linear constraints on $p(ab|xy)$, and we can
  conclude that $\cG$ is a polytope because $p(ab|xy)$ lies in
  $\NS$, \textit{i.e.} it is a valid behaviour, and lies in a polytope
  already. Dropping that hypothesis would leave a polyhedron but not a
  polytope, since game-type constraint allow for arbitrary negative
  coefficients on vectors fulfilling these games.
\end{remark}

\begin{remark}
  The game polytope could alternatively be defined directly by a finite list of games:
  positivity-type games, no-signaling games, and normalization games (which for us are trivial games, as we work in the $\NS$ space)
  and, for each predicate $V$, the facets of $\cG_V$, which are all
  games. Since any finite and complete list of valid inequalities for
  a polytope contains its facets, all facets of $\cG$ are games.
\end{remark}

\begin{lemma} \label{lem:polyhedron}
  Let $\{v_k\}$ be a finite set of vectors in $\R^n$ and let
  \begin{equation*}
    D = \conv\bigl(\{v_k\}\bigr) - \R^n_{\geq 0} .
  \end{equation*}
  Then $D$ is a polyhedron, and every valid inequality for $D$ has
  nonnegative coefficients.
\end{lemma}

\begin{proof}
  The set can be written as the sum of a convex polytope and a
  finitely generated polyhedral cone,
  \begin{equation*}
    D = P_g + U_- ,
  \end{equation*}
  where $P_g = \conv\bigl(\{v_k\}\bigr)$ and
  $U_- = \R^n_{\leq 0} = \operatorname{cone}\bigl(\{-\bm{e}_i\}\bigr)$,
  with $\bm{e}_i$ the canonical basis. A set is a polyhedron iff it
  can be written as the sum of a convex polytope and a polyhedral cone
  (see Ziegler~\cite[Theorem 1.2]{Ziegler1995}); hence $D$ is a
  polyhedron, i.e.\ defined by finitely many linear inequalities.

  Moreover any valid inequality on $D$ is of the form
  $\ip{f}{x} \leq b$ with $f \geq 0$: if some coefficient were
  negative, say $f_1 < 0$, then $f_1 x_1$ would be unbounded above,
  since $x_1$ can be made arbitrarily negative.
\end{proof}

\section{Explicit quantum violations of $J_\star$, $F_\star^1$ and $F_\star^2$}
\label{app:quantum_violations}

In this Appendix we give for each of the three game-inequivalent facets $J_\star$, $F_\star^1$ and $F_\star^2$ an explicit quantum behaviour violating it. By \cref{lemma:qpoint} each of these provides a proof for the existence of nonlocal quantum behaviours in $\mathcal{G}$ (see proof in Appendix \ref{app:quantumNL_in_G}). All constructions use the two-qubit maximally entangled state:

\begin{equation}
 \ket{\Phi^+} = \tfrac{1}{\sqrt{2}}\left( \ket{00} + \ket{11} \right),
\end{equation}

and projective measurements in the $X-Z$ plane, that is with elements of the form

\begin{equation}
  P_{\pm}(\theta) = \frac{ \1 \pm (\cos\theta \, Z + \sin\theta \, X) }{2}
\end{equation}
  
where $X$ and $Z$ are qubit Pauli matrices and we will associate $P_+(\theta)$ (resp.\ $P_-(\theta)$)  with output $0$ (resp. $1$). The only parameters are Alice's angles $\theta_x$ and Bob's angles $\phi_y$. Using state $\ket{\Phi^+} $ the behaviour in Collins-Gisin reads:

\begin{equation}
  p(00|xy) = \frac{1+\cos(\theta_x-\varphi_y)}{4},
  \qquad p_A(0|x) = p_B(0|y) = \frac12 .
  \label{eq:p00_phiplus}
\end{equation}

\subsection{Violation of $J_\star$ in $(4,4,2,2)$}

Remind that $J_\star$ in Collins-Gisin is given by 
\begin{equation}
\begin{split}
J_\star={}&(3,3,3,3;\,3,-1,-1,-2;\,1,1,-3,2;\,1,-3,1,2)\\
&\mid(-5,1,0,0)\mid(-5,0,0,-3)\leq1.
\end{split}
\end{equation}

With the three groups separated by vertical lines multiply, respectively, $p(00|xy)$, Alice's marginals, and Bob's marginals. 

\paragraph{Strategy.}

Alice and Bob share $\lvert\Phi^+\rangle$ and choose angles which are all multiples of $\pi/3$:

\begin{align*}
  x=0:&\ 0, & x=1:&\ 2\pi/3, & x=2:&\ \1, & x=3:&\ \pi/3,\\
  y=0:&\ \pi/3, & y=1:&\ -\pi/3, & y=2:&\ 0, & y=3:&\ 0,
\end{align*}

except Alice's measurement for $x=2$, which is the trivial POVM $\{\1,0\}$. Bob performs the same measurement for $y=2$ and $y=3$.

\paragraph{Check.}
From Eq.~\eqref{eq:p00_phiplus}, the correlators $\cos(\theta_x-\varphi_y)$ for $x\in\{0,1,3\}$ take only the values $\pm1,\pm\tfrac12$, while for $x=2$ one has $p(00|2y) = p_B(0|y) = \tfrac12$. This gives

\begin{equation}
  \big[p(00|xy)\big]_{x,y} =
  \begin{pmatrix}
    3/8 & 3/8 & 1/2 & 1/2\\
    3/8 & 0 & 1/8 & 1/8\\
    1/2 & 1/2 & 1/2 & 1/2\\
    1/2 & 1/8 & 3/8 & 3/8
  \end{pmatrix}.
\end{equation}

Summing row by row,

\begin{equation}
  \sum_{x,y} C_{xy}\, p(00|xy) = \frac{21}{4} + \frac{3}{4} + \frac{1}{2} + \frac{5}{4} = \frac{31}{4}.
\end{equation}
The marginals are $p_A(0|0)=p_A(0|1)=\tfrac12$ (and $p_A(0|2)=1$, which
does not appear in $J_\star$), and $p_B(0|0)=p_B(0|3)=\tfrac12$. Therefore
\begin{equation}
  J_\star(p) = \frac{31}{4} - \frac{5}{2} + \frac{1}{2} - \frac{5}{2} - \frac{3}{2}
  = \frac{7}{4} \;>\; 1 .
\end{equation}

\subsection{Violation of $F_\star^1$ and $F_\star^2$ in $(2,2,4,4)$}

We give again both inequalities here in Collins-Gisin notation: 

\begin{equation}
\begin{split}
F^1_{\star}={}&(3,2,1,2,2,1,1,1,1;\,3,2,1,2,1,1,1,1,0;\\
&\;\;3,2,1,2,1,1,1,1,0;\,-3,-2,-1,-2,-1,0,-1,0,0)\\
&\mid(-3,-2,-1;\,0,0,0)\mid(-3,-2,-1;\,0,0,0)\leq 0.
\end{split}
\label{eq:F1-CG_app}
\end{equation}

\begin{equation}
\begin{split}
F^2_{\star}={}&(3,2,2,2,2,1,2,1,0;\,3,2,1,2,2,1,2,2,2;\\
&\;\;3,2,2,2,2,2,1,1,2;\,-3,-2,-2,-2,-2,-1,-2,-1,0)\\
&\mid(-3,-2,-2;\,0,0,0)\mid(-3,-2,-2;\,0,0,0)\leq 0.
\end{split}
\label{eq:F2-CG_app}
\end{equation}

As for $J_{\star}$, the three groups separated by vertical lines multiply, respectively, $p(ab|xy)$, Alice's marginals $p_A(a|x)$, and Bob's marginals $p_B(b|y)$, with $a,b\in\{0,1,2\}$ and $x,y\in\{0,1\}$. In the first group, the blocks separated by semicolons correspond to $(x,y)=(0,0),(0,1),(1,0),(1,1)$; within each block the entries are ordered $ab=00,01,02,10,\dots,22$. In the marginal groups, semicolons separate $x=0$ from $x=1$ (respectively $y$).

Both facets can be violated by the \emph{same} behaviour, owing to the following observation. Consider a behaviour in which the outputs $1$ and $2$ never occur, i.e.\ $p(ab|xy)=0$ whenever $a\in\{1,2\}$ or $b\in\{1,2\}$. In Collins-Gisin notation, where output $3$ is implicit, every term of \cref{eq:F1-CG_app} and \cref{eq:F2-CG_app} involving the outputs $1$ or $2$ then vanishes, and only the coefficients of $p(00|xy)$, $p_A(0|x)$ and $p_B(0|y)$ survive. These are identical for the two inequalities: the first entry of each block is $3,3,3,-3$, and the marginal coefficients are $-3$ for $x=0$ (resp.\ $y=0$) and $0$ for $x=1$ (resp.\ $y=1$). Hence, on such behaviours,
\begin{align*}
  &F_\star^1(p) = F_\star^2(p) \\
  &= 3\Big[\, p(00|00) + p(00|01) + p(00|10) - p(00|11) - p_A(0|0) - p_B(0|0) \Big] \\
  &\;\overset{\mathcal L}{\le}\; 0,
\end{align*}
where we recognise $[\, p(00|00) + p(00|01) + p(00|10) - p(00|11) - p_A(0|0) - p_B(0|0) \Big]\leq0$ as the CHSH Bell inequality written in Collins-Gisin notation, also referred to it as the CH inequality~\cite{clauser1974experimental}.

\paragraph{Strategy.}

Alice and Bob share $\lvert\Phi^+\rangle$ and perform the standard Tsirelson measurements

\begin{align*}
  x=0:\ Z, \qquad & x=1:\ X, \\
  y=0:\ \frac{Z+X}{\sqrt2}, \qquad & y=1:\ \frac{Z-X}{\sqrt2},
\end{align*}

 outputting $0$ for the eigenvalue $+1$ and $3$ for the eigenvalue $-1$. As four-output POVMs, the elements associated with the outputs $1$ and $2$ are zero.

\paragraph{Check.}
From \cref{eq:p00_phiplus}, $p(00|xy) = \big(1+\tfrac{1}{\sqrt2}\big)/4$ for $(x,y)\neq(1,1)$, $p(00|11)=\big(1-\tfrac{1}{\sqrt2}\big)/4$, and $p_A(0|0)=p_B(0|0)=\tfrac12$. Therefore

\begin{align*}
     & F_\star^1(p) = F_\star^2(p) = \\
  &  3\left[ \frac{2+2\sqrt2}{4} - 1 \right]
  = \frac{3\left(\sqrt2-1\right)}{2} \approx 0.62 \;>\; 0 .
\end{align*} 

\section{Nonlocal quantum behaviours in $\cG$} \label{app:quantumNL_in_G}

Here we prove the existence of nonlocal quantum behaviours in $\cG$. We prove that these exist whenever a facet of $\cL$ admits a quantum violation, and is not a facet of $\cG$. This is thus the case for example with $J_{\star}$, $F_{\star}^1$ and $F_{\star}^2$. 

Let us consider a facet $F$ of $\cL$ and the following behaviour $p_{\star}$:

\begin{equation}
    p_{\star} = \sum_k c_k v_k
\end{equation}

where $c_k > 0, \sum_k c_k = 1$, and $v_k$ the vertices of $\cL$ saturating $F$. Since $c_k > 0 \; (\forall k)$, $p_{\star}$ is in the relative interior of facet $F$. 

Since $p_{\star}$ is in $\cL$, it is in $\cG$. Either (1) $p_{\star}$ is strictly inside of $\cG$, or (2) it lies on the boundary of $\cG$. If $F$ is not a facet of $\cG$, we can show that (2) is impossible: 

\begin{equation*}
    \begin{aligned}
    \text{If $p_{\star} \in \partial\cG$ there is a facet of $\cG$, $H x \leq b$, with } H p_{\star} = b \\
    \Rightarrow  H p_{\star} = \sum_k c_k H v_k = b \\
    \text{Since $\cL \subseteq \cG $, $H x \leq b$ is valid for $\cL$ too } \Rightarrow H v_k \leq b  \; \forall k \\
    \sum_k c_k H v_k = b \Leftrightarrow  H v_k = b  \; \forall k \\
    \end{aligned}
\end{equation*}

Thus, the facet-defining inequality $H x \leq b$ is saturated on \textit{all} vertices of the facet $F$: $H$ defines the same facet as $F$, a contradiction (since $H$ is a facet of $\cG$, thus a game, and $F$ is by assumption a game-inequivalent facet). That is, $p_{\star}$ is strictly inside of $\cG$. 

This allows us to construct nonlocal quantum behaviours in $\cG$. Let us now consider as an example inequality $J_{\star}$ and a quantum behaviour $p_Q$ that violates its local bound:

\[
J_{\star} \cdot p_Q > 1 
\]

and define for $0 \leq q \leq 1$

\[
p_T(q) = q \cdot p_Q + (1-q) \cdot p_{\star} 
\]

which thus fulfils for any $q > 0$

\[
J_{\star} \cdot p_T(q) > 1 .
\]

We thus have for any $0 < q \leq 1$

\begin{equation*}
    \begin{aligned}
        & p_T(q) \in \cQ \; \text{(by convexity) } \\
        & p_T(q) \centernot{\in} \cL   \text{(since it violates $J_{\star}$)}    
    \end{aligned}
\end{equation*}

Finally, since $p_{\star} $ is strictly in $\cG$, there exists a $q^* > 0$ such that $p_T(q^*) \in \cG$. Therefore, $p_T(q^*)$ is a quantum nonlocal behaviour in $\cG$. 

\subsection{Constructing a MILP to compute the $\cG$-visibility} \label{app:G_membership_milp}

Here we construct a MILP for the behaviour-membership in $\cG$. The idea is to write the maximization of the score for a fixed behaviour, minus the bound of the game, over all games. Consider the following optimization problem, defined for an inputted fixed $p(ab|xy)$:

\begin{equation} 
\begin{aligned}
c^* = \max_{G,\,V,\,\mu,\,\beta}\;\; & \sum_{abxy} G_{ab|xy}\, p(ab|xy) - \beta \\
\text{s.t.}\;\; & \beta \geq \sum_{abxy} G_{ab|xy}\, D^{(k)}(ab|xy) \quad \forall k, \\
& 0 \leq \mu_{xy} \leq 1 \quad \forall x,y,  \\
&\sum_{x,y} \mu(x,y) = 1  \\
& V_{ab|xy} \in \{0,1\} \quad \forall a,b,x,y, \\
& 0 \leq G_{ab|xy} \leq V_{ab|xy} \quad \forall a,b,x,y, \\
& \mu_{xy} - (1 - V_{ab|xy}) \leq G_{ab|xy}  \leq  \mu_{xy} + (1 - V_{ab|xy}) \\ 
& \quad \forall a,b,x,y,
\end{aligned}
\end{equation}
where $D^{(k)}(ab|xy)$ are the local deterministic behaviours of that scenario, i.e. the vertices of $\cL$. 

First, note that this optimization problem is a MILP: the objective functions and the constraints are linear with respect to the variables $G_{ab|xy}, V(a,b|x,y), \mu(x,y), \beta$. Next, note that $G_{ab|xy}$ is forced to be of the form $G_{ab|xy} =  \mu(x,y)V(a,b|x,y)$, thanks to the last two constraints: whenever $V(a,b|x,y) = 0$ we have that $G_{ab|xy} = 0$ from condition $G_{ab|xy} \leq V(a,b|x,y)$, and whenever $V(a,b|x,y) = 1$ we have that $G_{ab|xy} = \mu(x,y)$ from constraint $\mu(x,y) - (1-V(a,b|x,y) \leq  G_{ab|xy} \leq \mu(x,y) + (1 - V(a,b|x,y))$. We impose $V(a,b|x,y)  \in \{0,1\} $, that is, $V(a,b|x,y)$ is a valid predicate, while $ 0 \leq \mu(x,y) \leq 1, \sum_{x,y} \mu(x,y) = 1$ defines a valid prior. That is, $G_{ab|xy}$ is a game. 

Therefore, if $p(ab|xy)$ violates a game, the MILP will return $c^* > 0$, since this is already the case using that violated game as a feasible point, and similarly, getting a $c^* > 0$ explicitly provides a game violation for $p(ab|xy)$. Finally, if the MILP return $c^* = 0$, we can thus conclude that $p(ab|xy)$ cannot violate any game, and is in $\cG$. 

\subsubsection{Estimating the critical $\cG$-visibility of $p_T(q)$}

We can now use a behaviour from Appendix \ref{app:quantum_violations}, \textit{i.e.} a nonlocal quantum behaviour violating a game-inequivalent facet, to probe for nonlocal quantum behaviours in $\cG$. By the proof just above, we know that these exist and we can verify it numerically using the MILP presented in \eqref{eq:MILP_behaviour}. 

We consider game-inequivalent inequality $F_\star^1$ and we set 

\begin{equation}
    p_{\star} = \tfrac{1}{N} \sum_{k=1}^N v_k
\end{equation}

where $\{v_k\}_{k=1..N}$ are the vertices of $\cL$ saturating  $F_\star^1$. We take the example of quantum violation given in Appendix \ref{app:quantum_violations}, that is:

\begin{align*}
    p_Q(00|xy) = \big(1+\tfrac{1}{\sqrt2}\big)/4 \;\; \text{for} \, (x,y)\neq(1,1) \\ 
p_Q(00|11)=\big(1-\tfrac{1}{\sqrt2}\big)/4,  \\
p_Q^A(0|0)=p_Q^B(0|0)=\tfrac12. 
\end{align*}

And we define $p_T(q) = q \cdot p_Q + (1-q) \cdot p_{\star} $. Running the MILP \eqref{eq:MILP_behaviour} we find that $p_T(0.042)$ violates a nonlocal game, and is thus a quantum point outside of $\cG$, while $p_T(0.041)$ return $c^*=0$ and thus does not violate any nonlocal game, up to the solver precision. We give the explicit quantum realization of the behaviour as well as the MILP algorithm to reproduce these results in our repository \cite{github_game_vs_inequ}.  

\end{document}